\documentclass{article}
\usepackage[utf8]{inputenc}
\usepackage{amsmath}
\usepackage{amssymb}
\usepackage{enumerate}
\usepackage{todonotes}
\usepackage{amsthm}
\usepackage{subcaption} 
\usepackage{fullpage}
\usepackage[nocompress]{cite}

\newtheorem{theorem}{Theorem}

\newtheorem{lemma}[theorem]{Lemma}
\newtheorem{conjecture}{Conjecture}
\newcommand{\occ}[0]{\mathrm{occ}}
\newcommand{\str}[1]{\textrm{str}(#1)}
\newcommand{\loc}{\textrm{locus}}
\newcommand{\ors}{\textsf{RangeSuccessor}}
\newcommand{\orp}{\textsf{RangePredecessor}}
\newcommand{\yes}{\textsc{Yes}}
\newcommand{\no}{\textsc{No}}
\newcommand{\findi}[1]{\textsf{FindConsecutive}_{P_2}(#1)} 
\newcommand{\findj}[1]{\textsf{FindConsecutive}_{P_1}(#1)}

\newcommand{\qexistsab}[0]{\textsf{Exists}}
\newcommand{\qcountab}[0]{\textsf{Count}}
\newcommand{\qreportab}[0]{\textsf{Report}}

\title{Gapped Indexing for Consecutive Occurrences}
\author{Philip Bille \\ \texttt{phbi@dtu.dk} \and Inge Li G{\o}rtz \\ \texttt{inge@dtu.dk} \and Max Rish{\o}j Pedersen \\ \texttt{mhrpe@dtu.dk} \and Teresa Anna Steiner \\ \texttt{terst@dtu.dk}}
\date{}

\begin{document}
\maketitle

\begin{abstract}
The classic string indexing problem is to preprocess a string $S$ into a compact data structure that supports efficient pattern matching queries. Typical queries include \emph{existential queries} (decide if the pattern occurs in $S$), \emph{reporting queries} (return all positions where the pattern occurs), and \emph{counting queries} (return the number of occurrences of the pattern). In this paper we consider a variant of string indexing, where the goal is to compactly represent the string such that given two patterns $P_1$ and $P_2$ and a gap range $[\alpha, \beta]$ we can quickly find the consecutive occurrences of $P_1$ and $P_2$ with distance in $[\alpha, \beta]$, i.e., pairs of occurrences immediately following each other and with distance within the range. We present data structures that use $\widetilde{O}(n)$ space and query time  $\widetilde{O}(|P_1|+|P_2|+n^{2/3})$ for existence and counting and $\widetilde{O}(|P_1|+|P_2|+n^{2/3}\occ^{1/3})$ for reporting. We complement this with a conditional lower bound based on the set intersection problem showing that any solution using $\widetilde{O}(n)$ space must use $\widetilde{\Omega}(|P_1| + |P_2| + \sqrt{n})$ query time. To obtain our results we develop new techniques and ideas of independent interest including a new suffix tree decomposition and hardness of a variant of the set intersection problem. 
\end{abstract}

\section{Introduction}\label{sec:Intro}
The classic string indexing problem is to preprocess a string $S$ into a compact data structure that supports efficient pattern matching queries. Typical queries include \emph{existential queries} (decide if the pattern occurs in $S$), \emph{reporting queries} (return all positions where the pattern occurs), and \emph{counting queries} (return the number of occurrences of the pattern). 
An important variant of this problem is the \emph{gapped string indexing problem}~\cite{bille2014substring,iliopoulos2009indexing,bader2016practical,caceres2020fast,bille2014string,lewenstein2011indexing, keller2007range}. 
Here, the goal is to compactly represent the string such that given two patterns $P_1$ and $P_2$ and a \emph{gap range} $[\alpha, \beta]$ we can quickly find occurrences of $P_1$ and $P_2$ with distance in $[\alpha, \beta]$. Searching and indexing with gaps is frequently used in computational biology applications~\cite{BB1994, HPFB1999, FG2008, HSSS2011, Myers1996, MM1993a, NR2003, bader2016practical,caceres2020fast, BGVW2012}.

Another variant is \emph{string indexing for consecutive  occurrences}~\cite{navarro2015reporting,DBLP:conf/fsttcs/BilleGPRS20}. Here, the goal is to compactly represent the string such that given a pattern $P$ and a gap range $[\alpha, \beta]$ we can quickly find \emph{consecutive occurrences} of $P$ with distance in $[\alpha, \beta]$, i.e., pairs of occurrences immediately following each other and with distance within the range. 

In this paper, we consider the natural combination of these variants that we call \emph{gapped indexing for consecutive occurrences}. Here, the goal is to compactly represent the string such that given two patterns $P_1$ and $P_2$ and a gap range $[\alpha, \beta]$ we can quickly find the consecutive occurrences of $P_1$ and $P_2$ with distance in $[\alpha, \beta]$.  

We can apply standard techniques to obtain several simple solutions to the problem. To state the bounds, let $n$ be the size of $S$. If we store the suffix tree for $S$, we can answer queries by searching for both query strings, merging the results, and removing all non-consecutive occurrences. This leads to a solution using $O(n)$ space and $\widetilde{O}(|P_1| + |P_2| + \occ_{P_1} +  \occ_{P_2})$ query time, where $\occ_{P_1}$ and $\occ_{P_2}$ denote the number of occurrences of $P_1$ and $P_2$, respectively\footnote{$\widetilde{O}$ and $\widetilde{\Omega}$ ignores polylogarithmic factors}. However, $\occ_{P_1} + \occ_{P_2}$ may be as large as $\Omega(n)$ and much larger than the size of the output.

Alternatively, we can obtain a fast query time in terms of the output at the cost of increasing the space to $\Omega(n^2)$. To do so, store for each node $v$ in the suffix tree the set of all consecutive occurrences $(i,j)$ where $i$ is a position below $v$ in a 2D range searching data structure organized by the lexicographic order of $j$ and their distance from descendants of $v$. To answer a query, we then perform a 2D range search in the structure corresponding to $P_1$ using the lexicographic range in the suffix tree defined by $P_2$ and the gap range. This leads to a solution for reporting queries using $\widetilde{O}(n^2)$ space and $\widetilde{O}(|P_1| + |P_2| + \occ)$ time, where $\occ$ is size of the output. For existence and counting, we obtain the same bound without the $\occ$ term.    



In this paper, we introduce new solutions that significantly improve the above time-space trade-offs. Specifically, we present data structures that use  $\widetilde{O}(n)$ space and query time  $\widetilde{O}(|P_1|+|P_2|+n^{2/3})$ for existence and counting and $\widetilde{O}(|P_1|+|P_2|+n^{2/3}\occ^{1/3})$ for reporting. We complement this with a conditional lower bound based on the set intersection problem showing that any solution using $\widetilde{O}(n)$ space must use $\widetilde{\Omega}(|P_1| + |P_2| + \sqrt{n})$ query time. To obtain our results we develop new techniques and ideas of independent interest including a new suffix tree decomposition and hardness of a variant of the set intersection problem.  

\subsection{Setup and Results}
Throughout the paper, let $S$ be a string of length $n$. Given two patterns $P_1$ and $P_2$ a  \emph{consecutive occurrence} in $S$ is a pair of occurrences $(i,j)$, $0 \leq i < j < |S|$ where $i$ is an occurrence of $P_1$ and $j$ an occurrence of $P_2$, such that no other occurrences of either $P_1$ or $P_2$ occurs in between. The \emph{distance} of a consecutive occurrence $(i,j)$ is $j-i$. Our goal is to preprocess $S$ into a compact data structure that given pattern strings $P_1$ and $P_2$ and a gap range $[\alpha, \beta]$ supports the following queries:  
\begin{itemize}
    \item $\qexistsab(P_1, P_2, \alpha, \beta)$: 
    determine if there is a consecutive occurrence of $P_1$ and $P_2$ with distance within the range $[\alpha, \beta]$.
    \item $\qcountab(P_1, P_2, \alpha, \beta)$: return the number of consecutive occurrences of $P_1$ and $P_2$ with distance within the range $[\alpha, \beta]$. 
    \item $\qreportab(P_1, P_2, \alpha, \beta)$:
    report all consecutive occurrences of $P_1$ and $P_2$ with distance within the range $[\alpha, \beta]$.
\end{itemize}
We present new data structures with the following bounds: 
\begin{theorem}\label{thm:main-result} Given a string of length $n$, we can 
    \begin{enumerate}[(i)]
        \item \label{lem:ors,linear} construct an $O(n)$ space data structure that supports $\qexistsab(P_1, P_2, \alpha, \beta)$ and $\qcountab(P_1, P_2, \alpha, \beta)$ queries in $O(|P_1| + |P_2| + n^{2/3}\log^{\epsilon}n)$ time for constant $\epsilon>0$, or
        \label{thm:ab-existence-and-count}
        \item \label{lem:ors,fast} construct an $O(n \log n)$ space data structure that supports   $\qreportab(P_1, P_2, \alpha, \beta)$ queries in ${O}(|P_1| + |P_2| +  n^{2/3} \occ^{1/3}\log n\log \log n)$ time, where $\occ$ is the size of the output.
        \label{thm:ab-reporting}
    \end{enumerate}
\end{theorem}  
Hence, ignoring polylogarithmic factors, Theorem~\ref{thm:main-result} achieves $\widetilde{O}(n)$ space and query time $\widetilde{O}(|P_1| + |P_2|+n^{2/3})$ for existence and counting and $\widetilde{O}(|P_1| + |P_2|+n^{2/3}\occ^{1/3})$ for reporting. Compared to the above mentioned simple suffix tree approaches that finds all occurrences of the query strings and merges them, we match the $\widetilde{O}(n)$ space, while reducing the dependency on $n$ in the query time from worst-case $\Omega(|P_1| + |P_2|+n)$ to $\widetilde{O}(|P_1| + |P_2|+n^{2/3})$ for $\qexistsab$ and $\qcountab$ queries and $\widetilde{O}(|P_1| + |P_2|+n^{2/3}\occ^{1/3})$ for $\qreportab$ queries. 

We complement Theorem~\ref{thm:main-result} with a conditional lower bound based on the set intersection problem. Specifically, we use the Strong SetDisjointness Conjecture from~\cite{DBLP:conf/wads/GoldsteinKLP17} to obtain the following result: 
\begin{theorem}\label{thm:lower_bound}
    Assuming the Strong SetDisjointness Conjecture, any data structure on a string $S$ of length $n$ that supports \qexistsab\ queries in $O(n^{\delta}+|P_1|+|P_2|)$ time, for $\delta\in[0,1/2]$, requires  $\widetilde{\Omega}\left(n^{2-2\delta - o(1)}\right)$ space. This bound also holds if we limit the queries to only support ranges of the form $[0,\beta]$, and even if the bound $\beta$ is known at preprocessing time.
\end{theorem}
With $\delta = 1/2$, Theorem~\ref{thm:lower_bound} implies that any near linear space solution must have query time $\widetilde{\Omega}(|P_1| + |P_2| + \sqrt{n})$. Thus, Theorem~\ref{thm:main-result} is optimal within a factor roughly $n^{1/6}$. On the other hand, with $\delta = 0$, Theorem~\ref{thm:lower_bound} implies that any solution with optimal $\widetilde{O}(|P_1| + |P_2|)$ query time must use $\widetilde{\Omega}(n^{2-o(1)})$ space. Note that this matches the trade-off achieved by the above mentioned simple solution that combines suffix trees with  two-dimensional range searching data structures. 
 
Finally, note that Theorem~\ref{thm:lower_bound} holds even when the gap range is of the form $[0, \beta]$. As a simple extension of our techniques, we show how to improve our solution from Theorem~\ref{thm:main-result} to match Theorem~\ref{thm:lower_bound} in this special case.

 
    
\subsection{Techniques}
To obtain our results we develop new techniques and show new interesting properties of consecutive occurrences. We first consider $\qexistsab$ and $\qcountab$ queries. The key idea is to split gap ranges into large and small distances. For large distances  that there can only be a limited number of consecutive occurrences and we show how these can be efficiently handled using a segmentation of the string. For small distances, we cluster the suffix tree and store precomputed answers for selected pairs of nodes. Since the number of distinct distances is small we obtain an efficient bound on the space. 

We extend our solution for $\qexistsab$ and $\qcountab$ queries to handle $\qreportab$ queries. To do so we develop a new decomposition of suffix trees, called the \emph{induced suffix tree decomposition} that recursively divides the suffix tree in half by index in the string. Hence, the decomposition is a balanced binary tree, where every node stores the  suffix tree of a substring of $S$. We show how to traverse this structure to efficiently recover the consecutive occurrences. 

For our conditional lower bound we show a reduction based on the set intersection problem. Along the way we show that set intersection remains hard even if all elements in the instance have the same frequency.




    
\subsection{Related Work}
As mentioned, string indexing for gaps and consecutive occurrences are the most closely related lines of work to this paper. Another related area is \emph{document indexing}, where the goal is to preprocess a collection of strings, called \emph{documents}, to report those documents that contain patterns subject to various constraints. 
For a comprehensive overview of this area see the survey by Navarro~\cite{navarro2014spaces}.

A well studied line of work within document indexing is \emph{document indexing for top-$k$} queries ~\cite{munro2020ranked,shah2013top,hon2014space,biswas2018ranked,hon2010efficient,navarro2017time,Tsur13,Hon2013indexes,Hon2013faster,NavarroT14,MunroNNST17}. The goal is to efficiently report the top-$k$ documents of smallest weight, where the weight is a function of the query. 
Specifically, the weight can be the distance of a pair of occurrences of the same or two different query patterns ~\cite{hon2014space,navarro2017time,shah2013top,MunroNNST17}. 
The techniques for top-$k$ indexing (see e.g. Hon~et~al.~\cite{hon2014space}) can be adapted to efficiently solve gapped indexing for consecutive occurrences in the special case when the gap range is of the form $[0, \beta]$. However, since these techniques heavily exploit that the goal is to find the top-$k$ \emph{closest occurrences}, they do not generalize to general gap ranges.





There are several results on conditional lower bounds for pattern matching and string indexing \cite{larsen2015hardness,DBLP:conf/wads/GoldsteinKLP17,DBLP:conf/icalp/AmirCLL14,DBLP:conf/soda/KopelowitzPP16,DBLP:conf/isaac/AmirKLPPS16}. Notably, Ferragina et al. \cite{DBLP:journals/jcss/FerraginaKMS03} and Cohen and Porat \cite{DBLP:journals/tcs/CohenP10} reduce the \emph{two dimensional substring indexing problem} to set intersection (though the goal was to prove an upper, not a lower bound). In the two dimensional substring indexing problem the goal is to preprocess pairs of strings such that given two patterns we can output the pairs that contain a pattern each. Larsen et al. \cite{larsen2015hardness} prove a conditional lower bound for the document version of indexing for two patterns, i.e., finding all documents containing both of two query patterns. Goldstein et al. \cite{DBLP:conf/wads/GoldsteinKLP17} show that similar lower bounds can be achieved via conjectured hardness of set intersection. Thus, there are several results linking indexing for two patterns and set intersection. Our reduction is still quite different, since we need a translation from intersection to distance. 



\subsection{Outline}
The paper is organized as follows.
In Section~\ref{sec:Prelims} we define notation and recall some useful results.
In Section~\ref{sec:ab-existence} we show how to answer $\qexistsab$ and $\qcountab$ queries, proving Theorem~\ref{thm:main-result}(\ref{thm:b-existence-and-count}).
In Section~\ref{sec:ab-reporting} we show how to answer $\qreportab$ queries, proving Theorem~\ref{thm:main-result}(\ref{thm:ab-reporting}).
In Section~\ref{sec:lower-bound} we prove the lower bound, proving Theorem~\ref{thm:lower_bound}.
Finally, in Section~\ref{sec:leq-b-reporting} we apply our techniques to solve the variant where $\alpha = 0$.


\section{Preliminaries}\label{sec:Prelims}

\paragraph{Strings.}
A \emph{string} $S$ of length $n$ is a sequence $S[0]S[1]\dots S[n-1]$ of characters from an alphabet $\Sigma$. A contiguous subsequence $S[i,j]=S[i]S[i+1]\dots S[j]$ is a \emph{substring} of $S$. The substrings of the form $S[i,n-1]$ are the \emph{suffixes} of $S$. 
The \emph{suffix tree}~\cite{weiner1973linear} is a compact trie of all suffixes of $S\$$, where \$ is a symbol not in the alphabet, and is lexicographically smaller than any letter in the alphabet. Each leaf is labelled with the index $i$ of the suffix $S[i,n-1]$ it corresponds to. Using perfect hashing~\cite{FKS1984}, the suffix tree can be stored in $O(n)$ space and solve the string indexing problem (i.e., find and report all occurrences of a pattern $P$) in $O(m+\occ)$ time, where $m$ is the length of $P$ and $\occ$ is the number of times $P$ occurs in $S$.

For any node $v$ in the suffix tree, we define $\str{v}$ to be the string found by concatenating all labels on the path from the root to $v$.
The \emph{locus} of a string $P$, denoted $\loc(P)$, is the minimum depth node $v$ such that $P$ is a prefix of $\str{v}$. 
The \emph{suffix array} stores the suffix indices of $S\$$ in lexicographic order. We identify each leaf in the suffix tree with the suffix index it represents. The suffix tree has the property that the leaves below any node represent suffixes that appear in consecutive order in the suffix array. For any node $v$ in the suffix tree, $\textrm{range}(v)$ denotes the range that $v$ spans in the suffix array. 
The \emph{inverse suffix array} is the inverse permutation of the suffix array, that is, an array where the $i$th element is the index of suffix $i$ in the suffix array.

\paragraph{Orthogonal range successor.}
The \emph{orthogonal range successor problem} is to  preprocess an array $A[0,\dots,n-1]$ into a data structure that  efficiently supports the following queries:
\begin{itemize}
    \item $\ors(a,b,x)$: return the successor of $x$ in $A[a,\dots, b]$, that is, the minimum $y> x$ such that there is an $i\in[a,b]$ with $A[i]=y$.
    \item $\orp(a,b,x)$: return the predecessor of $x$ in $A[a,\dots, b]$, that is, the maximum $y< x$ such that there is an $i\in[a,b]$ with $A[i]=y$.
\end{itemize}

\section{Existence and Counting}\label{sec:ab-existence}
In this section we give a data structure that can answer $\qexistsab$ and $\qcountab$ queries.
The main idea is to split the query interval into ``large" and ``small" distances.
For large distances we exploit that there can only be a small number of consecutive occurrences and we check them with a simple segmentation of $S$.
For small distances we cluster the suffix tree and precompute answers for selected pairs of nodes. 

We first show how to use orthogonal range successor queries to find consecutive occurrences.
Then we define the clustering scheme used for the suffix tree and give the complete data structure.

\subsection{Using Orthogonal Range Successor to Find Consecutive Occurrences}\label{subsec:cons_occ}
Assume we have found the loci of $P_1$ and $P_2$ in the suffix tree.
Then we can answer the following queries in a constant number of orthogonal range successor queries:
\begin{itemize}
    \item $\findi{i}$: given an occurrence $i$ of $P_1$, return the consecutive occurrence $(i,j)$ of $P_1$ and $P_2$, if it exists, and \no\ otherwise.
    \item $\findj{j}$: given an occurrence $j$ of $P_2$, return the consecutive occurrence $(i,j)$ of $P_1$ and $P_2$, if it exists, and \no\ otherwise.
\end{itemize}
Given a query $\findi{i}$, we answer as follows.
Compute $j = \ors(\textrm{range}(\loc(P_2)), i)$ to get the closest occurrence of $P_2$ after $i$.
Compute $i' = \orp(\textrm{range}(\loc(P_1)), j)$ to get the closest occurrence of $P_1$ before $j$.
If $i = i'$ then no other occurrence of $P_1$ exists between $i$ and $j$ and they are consecutive.
In that case we return $(i,j)$. Otherwise, we return $\no$. 

Similarly, we can answer $\findj{j}$ by first doing a \orp\ and then a \ors\ query.
Thus, given the loci of both patterns and a specific occurrence of either $P_1$ or $P_2$, we can in a constant number of $\ors$ and $\orp$ queries find the corresponding consecutive occurrence, if it exists.

\subsection{Data Structure}
To build the data structure we will use a cluster decomposition of the suffix tree.

\paragraph{Cluster Decomposition}
A cluster decomposition of a tree $T$ is defined as follows: For a connected subgraph $C\subseteq T$, a \emph{boundary node} $v$ is a node $v\in C$ such that either $v$ is the root of $T$, or $v$ has an edge leaving $C$ -- that is, there exists an edge $(v,u)$ in the tree $T$ such that $u\in T \setminus C$. 
A \emph{cluster} is a connected subgraph $C$ of $T$ with at most two boundary nodes.
A cluster with one boundary node is called a \emph{leaf cluster}. A cluster with two boundary nodes is called a \emph{path cluster}.
For a path cluster $C$, the two boundary nodes are connected by a unique path. We call this path the \emph{spine} of $C$. 
A \emph{cluster partition} is a partition of $T$ into clusters, i.e. a set $CP$ of clusters such that $\bigcup_{C\in CP}V(C)=V(T)$ and $\bigcup_{C\in CP}E(C)=E(T)$ and no two clusters in $CP$ share any edges. Here, $E(G)$ and $V(G)$ denote the edge and vertex set of a (sub)graph $G$, respectively. 
We need the next lemma which follows from well-known tree decompositions  ~\cite{alstrup1997minimizing, alstrup2000maintaining, alstrup2002improved, frederickson1997ambivalent} (see Bille and Gørtz~\cite{bille2011tree} for a direct proof).

\begin{lemma}\label{lem:cluster}
    Given a tree $T$ with $n$ nodes and a parameter $\tau$, there exists a cluster partition $CP$ such that $|CP|=O(n/\tau)$ and every $C\in CP$ has at most $\tau$ nodes. Furthermore, such a partition can be computed in $O(n)$ time.
\end{lemma}

\paragraph{Data Structure} We build a clustering of the suffix tree of $S$ as in Lemma~\ref{lem:cluster}, with cluster size at most $\tau$, where $\tau$ is some parameter satisfying $0 < \tau \leq n$.
Then the counting data structure consists of:
\begin{itemize} 
    \item The suffix tree of $S$, with some additional information for each node.
    For each node $v$ we store:
    \begin{itemize}
        \item The range $v$ spans in the suffix array, i.e., $\textrm{range}(v)$.
        \item A bit that indicates if $v$ is on a spine.
        \item If $v$ is on a spine, a pointer to the lower boundary node of the spine.
        \item If $v$ is a leaf, the local rank of $v$. That is, the rank of $v$ in the text order of the leaves in the cluster that contains $v$. Note that this is at most $\tau$.
    \end{itemize}
        \item The inverse suffix array of $S$.
    \item A range successor data structure on the suffix array of $S$.

    \item An array $M(u,v)$ of length $\lfloor\frac{n}{\tau}\rfloor + 1$ for every pair of boundary nodes $(u,v)$.
    For $1 \leq x \leq \lfloor\frac{n}{\tau}\rfloor$, 
    $M(u,v)[x]$ is the number of consecutive occurrences $(i,j)$ of $\str{u}$ and $\str{v}$ with distance at most $x$. We set $M(u,v)[0] = 0$. 
    

 
\end{itemize}

Denote $M(u,v)[\alpha,\beta]=M(u,v)[\beta]-M(u,v)[\alpha-1]$, that is, $M(u,v)[\alpha,\beta]$ is the number of consecutive occurrences of $\str{u}$ and $\str{v}$ with a distance in $[\alpha,\beta]$.

\paragraph{Space Analysis.}
We store a constant amount per node in the suffix tree.
The suffix tree and inverse suffix array occupy $O(n)$ space. For the orthogonal range successor data structure we  use the data structure of Nekrich and Navarro~\cite{DBLP:conf/swat/NekrichN12} which uses $O(n)$ space and $O(\log^\epsilon n)$ time, for constant $\epsilon > 0$.
There are $O\left(n^2/\tau^2\right)$ pairs of boundary nodes and for each pair we store an array of length $O\left(n/\tau\right)$. 
Therefore the total space consumption is $O\left(n + n^3/\tau^3\right)$.

\subsection{Query Algorithm}
We now show how to count the consecutive occurrences $(i,j)$ with a distance in the interval, i.e. $\alpha \leq j - i \leq \beta$. We call each such pair a \textit{valid occurrence}.

To answer a query we split the query interval $[\alpha, \beta]$ into two: $[\alpha,\lfloor\frac{n}{\tau}\rfloor]$ and $[\lfloor\frac{n}{\tau}\rfloor+1, \beta]$, and handle these separately. 

\subsubsection{Handling Distances $> \frac{n}{\tau}$.}
We start by finding the loci of $P_1$ and $P_2$ in the suffix tree. As shown in Section~\ref{subsec:cons_occ}, this allows us to find the consecutive occurrence containing a given occurrence of either $P_1$ or $P_2$.
We implicitly partition the string $S$ into segments of (at most) $\lfloor n / \tau\rfloor$ characters by calculating $\tau$ segment boundaries. Segment $i$, for $0 \leq i < \tau$, contains characters $S[i \cdot \lfloor\frac{n}{\tau}\rfloor, (i+1) \cdot \lfloor\frac{n}{\tau}\rfloor - 1]$ and segment $\tau$ (if it exists) contains the characters $S[\tau\cdot \lfloor\frac{n}{\tau}\rfloor,n-1]$. 
We find the last occurrence of $P_1$ in each segment by performing a series of $\orp$ queries, starting from the beginning of the last segment. 
Each time an occurrence $i$ is found we perform the next query from the segment boundary to the left of $i$, continuing until the start of the string is reached.
For each occurrence $i$ of $P_1$ found in this way, we use $\findi{i}$ to find the consecutive occurrence $(i,j)$ if it exists.
We check each of them, discard any with distance $\leq \frac{n}{\tau}$ and count how many are valid.

\subsubsection{Handling Distances $\leq \frac{n}{\tau}$.}
In this part, we only count valid occurrences with distance $\leq \frac{n}{\tau}$. Consider the loci of $P_1$ and $P_2$ in the suffix tree.  Let $C_i$ denote the cluster that contains $\loc(P_i)$ for $i=1,2$. There are two main cases.  

\paragraph{At least one locus is not on a spine} If either locus is in a small subtree hanging off a spine in a cluster or in a leaf cluster, we directly find all consecutive occurrences as follows: 
If $\loc(P_1)$ is in a small subtree then we use $\findi{i}$ on each leaf $i$ below $\loc(P_1)$ to find all consecutive occurrences, count the valid occurrences and terminate.
If only $\loc(P_2)$ is in a small subtree then we use $\findj{j}$ for each leaf $j$ below $\loc(P_2)$, count the valid occurrences and terminate.

\paragraph{Both loci are on the spine} If neither locus is in a small subtree then both are on a spine.
Let $(b_1, b_2)$ denote the lower boundary nodes of the clusters $C_1$ and $C_2$, respectively.
There are two types of consecutive occurrences $(i,j)$: 
\begin{enumerate}[(i)]
\item Occurrences where either $i$ or $j$ are inside $C_1$ resp. $C_2$.
\item Occurrences below the boundary nodes, that is, $i$ is below $b_1$ and $j$ is below $b_2$.
\end{enumerate}
See Figure~\ref{fig:count-query-and-false}(a). We describe how to count the different types of occurrences next.

\subparagraph{Type (i) occurrences}
To find the valid occurrences $(i,j)$ where either $i\in C_1$ or $j\in C_2$  we do as follows. 
First we find all the consecutive occurrences $(i,j)$ where $i$ is a leaf in $C_1$ by computing $\findi{i}$ for all leaves $i$  below $\loc(P_1)$ in $C_1$.  We count all valid occurrences we find in this way. 
Then we find all remaining consecutive occurrences $(i,j)$ where $j$ is a leaf in $C_2$ by computing $\findj{j}$ for all leaves $j$ below $\loc(P_2)$ in $C_2$. If $\findj{j}$ returns a valid occurrence $(i,j)$ we use the inverse suffix array to check if the leaf $i$ is below $b_1$. This can be done by checking whether $i$'s position in the suffix array is in $\textrm{range}(b_1)$.
If $i$ is below $b_1$ we count the occurrence, otherwise we discard it. 

\subparagraph{Type (ii) occurrences}
Next, we count the consecutive occurrences $(i,j)$, where both $i$ and $j$ are below $b_1$ and $b_2$, respectively. We will use the precomputed table, but we have to be a careful not to overcount.
By its construction, $M(b_1,b_2)[\alpha,\min(\lfloor\frac{n}{\tau}\rfloor, \beta)]$ is the number of consecutive occurrences $(i',j')$ of $\str{b_1}$ and $\str{b_2}$, where $\alpha \leq j' - i' \leq \min(\lfloor\frac{n}{\tau}\rfloor, \beta)$.
However, not all of these occurrence $(i',j')$ are necessarily \emph{consecutive} occurrences of $P_1$ and $P_2$, as there could be an occurrence of $P_1$ in $C_1$ or $P_2$ in $C_2$ which is between $i'$ and $j'$. 
We call such a pair $(i',j')$ a \textit{false occurrence}.
See Figure~\ref{fig:count-query-and-false}(b).
We proceed as follows.


\begin{figure}[t]
    \includegraphics[width=0.9\linewidth]{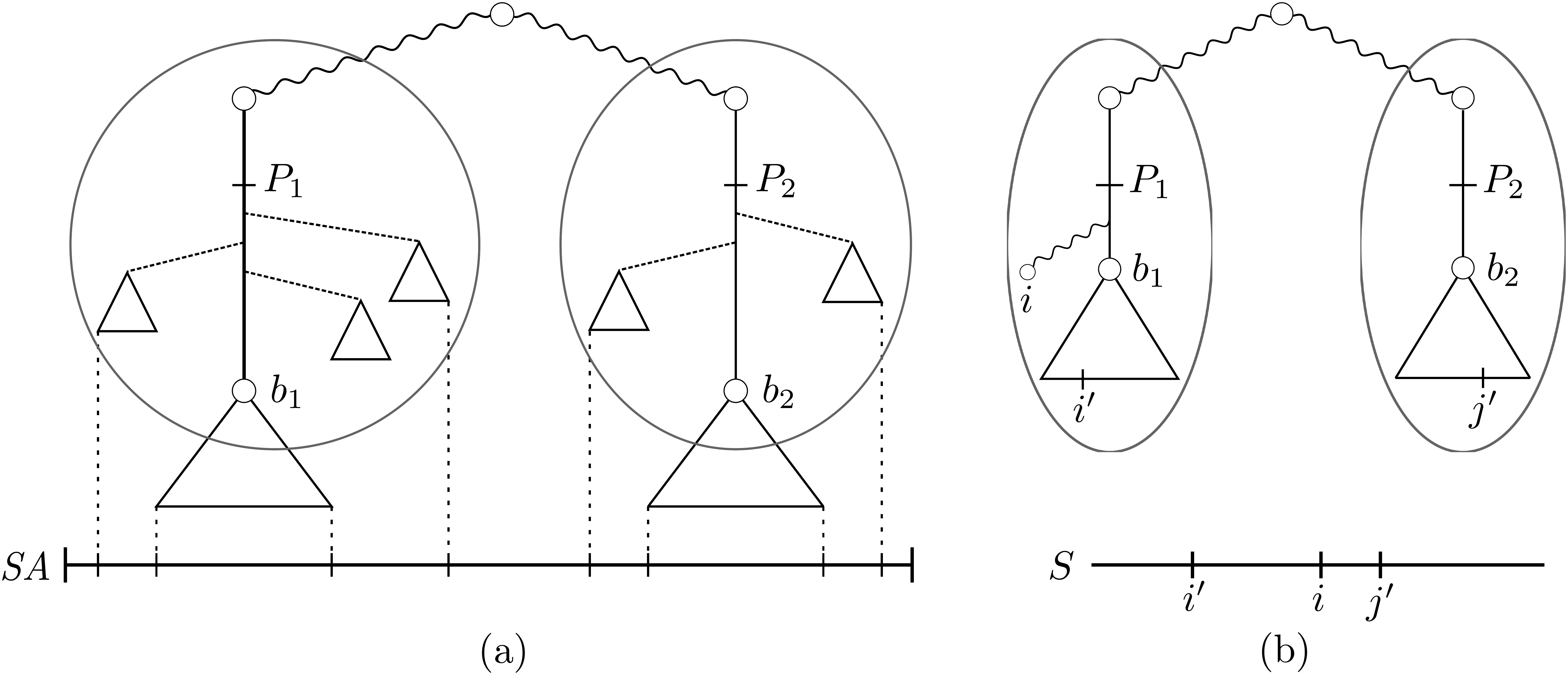}
    \caption{
    (a) Any consecutive occurrences $(i,j)$ of $P_1$ and $P_2$ is either also a consecutive occurrence of $\str{b_1}$ and $\str{b_2}$, or $i$ or $j$ are within the cluster. The suffix array is shown in the bottom with the corresponding ranges marked.
    (b) Example of a false occurrence. Here $(i',j')$ is a consecutive occurrence of $\str{b_1}$ and $\str{b_2}$, but not a consecutive occurrence of $P_1$ and $P_2$ due to $i$. The string $S$ is shown in bottom with the positions of the occurrences marked.
    }
    \label{fig:count-query-and-false}
\end{figure}

\begin{enumerate}
\item Set $c = M(b_1,b_2)[\alpha,\min(\lfloor\frac{n}{\tau}\rfloor, \beta)]$.
\item 

Construct the lists $L_i$ containing the leaves in $C_i$ that are below  $\loc(P_i)$ sorted by text order for $i = 1,2$. We can obtain the lists as follows. 
Let $[a,b]$ be the range of $\loc(P_i)$ and $[a',b']=\textrm{range}(b_i)$. Sort the leaves in $[a,a'-1] \cup [b'+1,b]$ using their local rank. 
 
\item Until both lists are empty iteratively pick and remove the smallest element $e$ from the start of either list. There are two cases. 
\begin{itemize}
\item $e$ is an element of $L_1$. 
\begin{itemize}
\item Compute $j' = \ors(\textrm{range}(b_2), e)$ to get the closest occurrence of $\str{b_2}$ after $e$. 
\item Compute $i' = \orp(\textrm{range}(b_1), j')$ to get the closest occurrence of $\str{b_1}$ before~$j'$. 
\end{itemize}
\item $e$ is an element of $L_2$. 
\begin{itemize}
\item Compute $i' = \orp(\textrm{range}(b_2), e)$ to get the previous occurrence $i'$ of $\str{b_1}$. 
\item Compute $j' = \ors(\textrm{range}(b_1), j')$ to get the following occurrence $j'$ of $\str{b_2}$. 
\end{itemize}
\end{itemize}

If $\alpha \leq j'-i' \leq \min(\lfloor\frac{n}{\tau}\rfloor, \beta)$ and $i' < e< j'$ decrement $c$ by  one. 
We
skip any subsequent occurrences that are also inside $(i',j')$. As the lists are sorted by text order, all occurrences that are within the same consecutive occurrence $(i',j')$ are handled in sequence.
\end{enumerate}

Finally, we add the counts of the different type of occurrences. 

\paragraph{Correctness.}
Consider a consecutive occurrence $(i,j)$ where $j-i > \frac{n}{\tau}$.
Such a pair must span a segment boundary, i.e., $i$ and $j$ cannot be in the same segment. 
As $(i, j)$ is a \textit{consecutive} occurrence, $i$ is the last occurrence of $P_1$ in its segment and $j$ is the first occurrence of $P_2$ in its segment. With the $\orp$ queries we find all occurrences of $P_1$ that are the last in their segment. 
We thus check and count all valid occurrences of large distance in the initial pass of the segments.

If either locus is in a small subtree we use $\findi{.}$ or $\findj{.}$ on the leaves below that locus, which by the arguments in Section~\ref{subsec:cons_occ} will find all consecutive occurrences. 

Otherwise, both loci are on a spine.
To count type~(i) occurrences we use $\findi{i}$ for all leaves $i$ below $\loc(P_1)$ in $C_1$ and $\findj{j}$ for all leaves $j$ below $\loc(P_2)$ in $C_2$.   
However, any valid occurrence $(i,j)$ where both $i \in C_1$ and $j \in C_2$ is found by both operations. Therefore, whenever we find a valid occurrence $(i,j)$ via $i=\findj{j}$ for $j \in C_2$, we only count the occurrence if $i$ is below $b_1$.
Thus we count all type~(i) occurrences exactly once.


To count type~(ii) occurrences we start with $c=M(b_1,b_2)[\alpha,\min(\lfloor\frac{n}{\tau}\rfloor, \beta)]$, which is the number of consecutive occurrences $(i',j')$ of $\str{b_1}$ and $\str{b_2}$, where $\alpha \leq j' - i' \leq \min(\lfloor\frac{n}{\tau}\rfloor, \beta)$.
Each $(i',j')$ is either also a consecutive occurrence of $P_1$ and $P_2$, or there exists an occurrence of $P_1$ or $P_2$ between $i'$ and $j'$.
Let $(i',j')$ be a false occurrence and let w.l.o.g.\ $i$ be an occurrence of $P_1$ with $i'<i<j'$.
Then $i$ is a leaf in $C_1$, since $(i',j')$ is a \emph{consecutive} occurrence of $\str{b_1}$ and $\str{b_2}$.
In step 3 we check for each leaf inside the clusters below the loci, if it is between a consecutive occurrence $(i',j')$ of $\str{b_1}$ and $\str{b_2}$ and if $\alpha \leq j' - i' \leq \min(\lfloor\frac{n}{\tau}\rfloor,\beta)$.
In that case $(i',j')$ is a false occurrence and we adjust the count $c$.
As $(i',j')$ can have multiple occurrences of $P_1$ and $P_2$ inside it, we skip subsequent occurrences inside $(i',j')$. 
After adjusting for false occurrences, $c$ is the number of type~(ii) occurrences.

\paragraph{Time Analysis.}
We find the loci in $O(|P_1| + |P_2|)$ time. 
Then we perform a number of range successor and find consecutive queries.  The time for a find consecutive query is bounded by the time to do a constant number of range successor queries.
To count the large distances we check at most $\tau$ segment boundaries and thus perform $O(\tau)$ range successor and find consecutive  queries. If either locus is not on a spine we check the leaves below that locus. There are at most $\tau$ such leaves due to the clustering.  To count type~(i) occurrences we check the leaves below the loci and inside the clusters. There are at most $2\tau$ such leaves in total. To count type~(ii) occurrences we check two lists constructed from the leaves inside the clusters below the loci. There are again at most $2\tau$ such leaves in total. For each of these $O(\tau)$ leaves we use a constant number of range successor and find consecutive queries. Thus the time for this part is bounded  by the time to perform $O(\tau)$ range successor queries.

Using the data structure of Nekrich and Navarro~\cite{DBLP:conf/swat/NekrichN12}, each range successor query takes $O(\log^{\epsilon} n)$ time so the total time for these queries is $O(\tau \log^{\epsilon} n)$.
For type~(ii) occurrences we sort two lists of size at most $\tau$ from a universe of size $\tau$, which we can do in $O(\tau)$ time.
Thus, the total query time is $O(|P_1| + |P_2| + \tau \log^{\epsilon} n)$.

\bigskip

Setting $\tau=\Theta(n^{2/3})$ we get a data structure that uses $O\left(n + n^3/\tau^3\right) = O(n)$ space and has query time $O(|P_1| + |P_2| + \tau \log^{\epsilon} n) = O(|P_1| + |P_2| + n^{2/3} \log^{\epsilon} n)$, for constant $\epsilon > 0$.
Given an $\qexistsab$ query we answer with a $\qcountab$ query, terminating when the first valid occurrence is found.
This concludes the proof of Theorem~\ref{thm:main-result}(\ref{thm:ab-existence-and-count}).

\section{Reporting}\label{sec:ab-reporting}
In this section, we describe our data structure for reporting queries.
Note that in Section \ref{sec:ab-existence}, we explicitly find all valid occurrences \emph{except} for type~(ii) occurrences, where we use the precomputed values. In this section, we describe how we can use a recursive scheme to report these.

The main idea, inspired by fast set intersection by Cohen and Porat~\cite{DBLP:journals/tcs/CohenP10}, 
is to build a binary structure which allows us to recursively divide into subproblems of half the size. Intuitively, the subdivision is a binary tree where every node contains the suffix tree of a substring of $S$. 
We use this structure to find type~(ii) occurrences by recursing on smaller trees.
We define the binary decomposition of the suffix tree next.
The details of the full solution follow after that.

\subsection{Induced Suffix Tree Decomposition}\label{subsec:OSTD}
 Let $T$ be a suffix tree of a string $S$ of length $n$. For an interval $[a,b]$ of \emph{text positions}, we define $T[a,b]$ to be the subtree of $T$ \emph{induced} by the leaves in $[a,b]$: That is, we consider the subtree consisting of leaves in $[a,b]$ together with their ancestors. We then delete each node that has only one child in the subtree and contract its ingoing and outgoing edge. 
See Figure \ref{fig:induced_subtree}. 

\begin{figure}[t]
    \centering
    \includegraphics[width=0.8\textwidth]{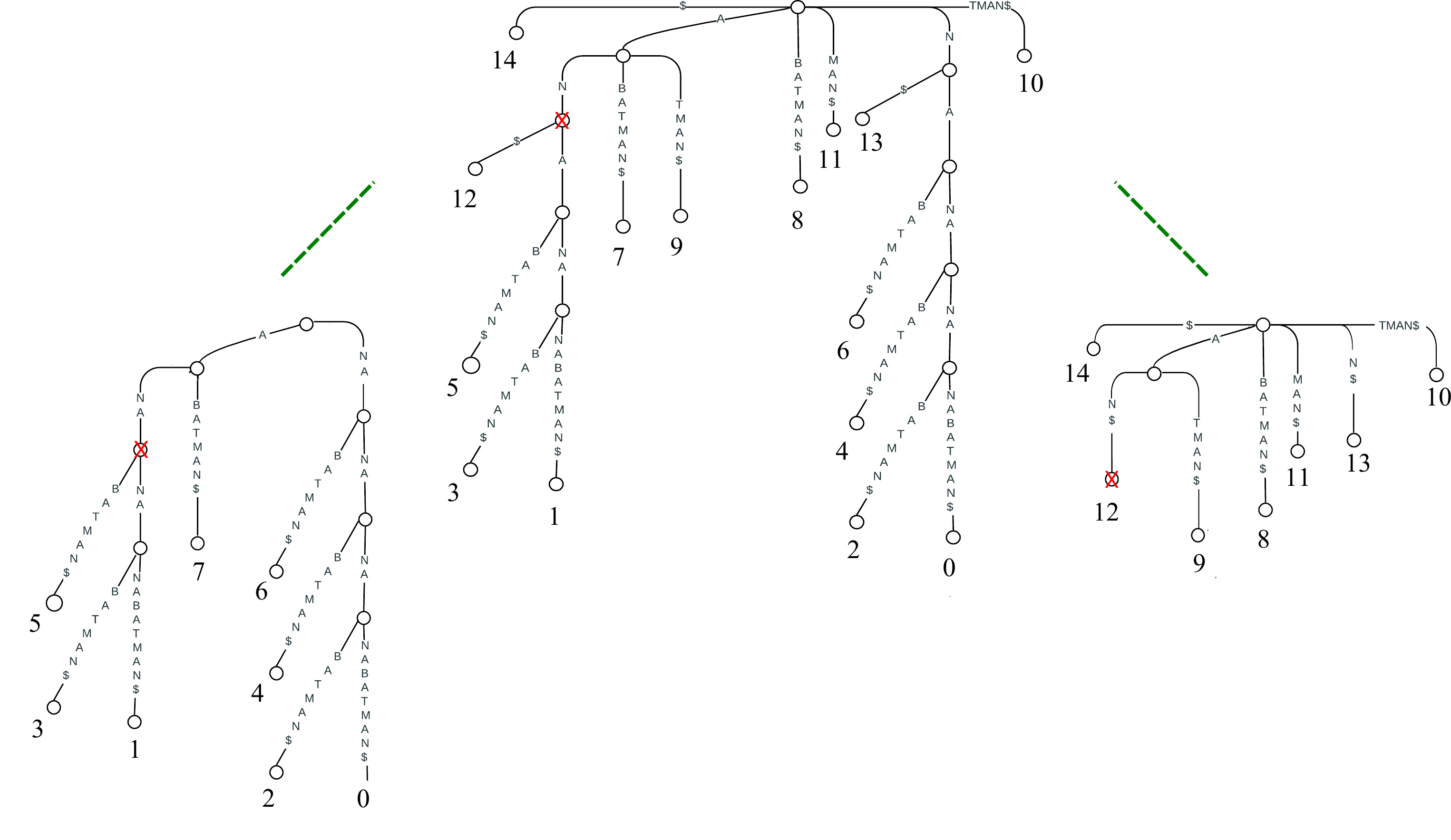}
    \caption{The suffix tree of \texttt{NANANANABATMAN\$} together with its children trees $T[0,7]$ and $T[8,14]$. The red crosses show a node in the parent tree and and its successor nodes in the two children trees.}
    \label{fig:induced_subtree}
\end{figure}

The \emph{induced suffix tree decomposition of $T$} now consists of
a higher level binary tree structure, the \emph{decomposition tree}, where each node corresponds to an induced subtree of the suffix tree. The root corresponds to  $T[0,n-1]$, and whenever we move down in the decomposition tree, the interval splits in half.  We also associate a level with each of the induced subtrees, which is their depth in the decomposition tree. In more detail, the decomposition tree is a binary tree such that:
\begin{itemize}
    \item The root of the decomposition tree corresponds to $T[0,n-1]$ and has level $0$.
    \item For each $T[a,b]$ of level $i$ in the decomposition, if $b-a>1$, its two children in the decomposition tree are $T[a,c]$ and  $T[c+1,b]$ where $c=\lfloor\frac{a+b}{2}\rfloor$; we will sometimes refer to these as ``children trees" to differentiate from children in the suffix tree.
\end{itemize}


The decomposition tree is a balanced binary tree and the total size of the induced subtrees in the decomposition is $O(n\log n)$: There are at most $2^i$ decomposition tree nodes on level $i$, each of which corresponds to an induced subtree of size $O\left(\frac{n}{2^i}\right)$
, and thus the total size of the trees on each of the $O(\log n)$ levels is~$O(n)$.

For each node $v$ in $T[a,b]$, we define the \emph{successor node} of $v$ in each of the children trees of $T[a,b]$ in the following way: If $v$ exists in the child tree, the successor node is $v$. Else, it is the closest descendant which is present. Note that from the way the induced subtrees are constructed, $v$ has at most one successor node in each child tree.

The induced suffix tree decomposition of $S$ consists of:
\begin{itemize}
    \item Each $T[a,b]$ stored as a compact trie.
    \item For each $T[a,b]$ we store a ``cropped" suffix array $SA_{[a,b]}$, that is, the suffix array of $S[a,b]$ with the original indices within $S$.
  \item For each node $v$ in $T[a,b]$ we store a pointer from $v$ to its successor nodes in each child tree, if it exists,
    and the interval in $SA_{[a,b]}$ that corresponds to the leaves below $v$. 
\end{itemize}
Since we store only constant information per node in any $T[a,b]$, the total space usage of this is $O(n\log n)$.

\subsection{Data Structure}
The reporting data structure consists of:
\begin{itemize}
    \item The induced suffix tree decomposition for $S$,
    \item An orthogonal range successor data structure on the suffix array, and
    \item The data structure from Section~\ref{sec:ab-existence} for each $T[a,b]$ in the induced suffix tree decomposition with parameters $n_i$ and $\tau_i$, where $n_i=\lfloor\frac{n}{2^i}\rfloor$ and $\tau_i=\Theta(n_i^{2/3})$, such that $n_i/\tau_i=\lfloor n_i^{1/3}\rfloor$. The only change is that we do not store an orthogonal range successor data structure for each of the induced subtrees. 
    \end{itemize}

\paragraph{Space Analysis.} 
We use the $O(n \log \log n)$ space and $O(\log \log n)$ time orthogonal range successor structure of Zhou~\cite{DBLP:journals/ipl/Zhou16}.
The existence data structure for each $T[a,b]$ of level $i$ is linear in $n_i$. 
Thus, by the arguments of Section \ref{subsec:OSTD}, the total space is $O(n\log n)$.

\subsection{Query Algorithm}

The main idea behind the algorithm is the following: For large distances, as in Section \ref{sec:ab-existence}, we implicitly segment $S$ to find all consecutive occurrences of at least a certain distance. For small distances, we are going to use the cluster decomposition and counting arrays to decide whether valid occurrences exist. That is, if one of the loci is in a small subtree, we use $\findi{.}$ resp. $\findj{.}$ to find all consecutive occurrences. Else, we perform a query as in Section \ref{sec:ab-existence} to decide whether any valid occurrences exist, and if yes, we recurse on smaller subtrees.

The idea here is, that in the induced suffix tree decomposition, the trees are divided in half by \emph{text position} - therefore, a \emph{consecutive} occurrence either will be fully contained in the left child tree, fully contained in the right child tree, or have the property that the occurrence of $P_1$ is the maximum occurrence in the left child tree and the occurrence of $P_2$ is the minimum occurrence in the right child tree. We will check the border case each time when we recurse.

 In detail, we do the following:
 We find the loci of $P_1$ and $P_2$ in the suffix tree.
    As in the previous section, we check $\tau_0$ segment boundaries with $\tau_0=\Theta(n^{2/3})$ to find all consecutive occurrences with distance within $[\max(\alpha,\lfloor n^{1/3}\rfloor),\beta]$. Now, we only have to find consecutive occurrences of distance within $[\alpha, \min(\beta, \lfloor n^{1/3}\rfloor)]$ in $T=T[0,n-1]$. In general, let $n_i=\lfloor\frac{n}{2^i}\rfloor$ and $\beta_i=\min(\beta, \lfloor n_i^{1/3}\rfloor)$ and let $T[a,b]$ be an induced subtree of level $i$. 
    
    To find all consecutive occurrences with distance within $[\alpha, \beta_i]$ in $T[a,b]$ of level $i$, 
    given the loci of $P_1$ and $P_2$ in $T[a,b]$, recursively do the following:
    \begin{itemize}
    \item If any of the loci is not on a spine of a cluster, we find all consecutive occurrences using $\findi{.}$ resp. $\findj{.}$ and check for each of them if they are valid; we report all such, then terminate.
    \item Else, we use the query algorithm for small distances from  Section~\ref{sec:ab-existence} to decide whether a valid occurrence with distance within $[\alpha,\beta_i]$ exists in $T[a,b]$.
    
    If such a valid occurrence exists, we recurse; that is,
 set $c=\lfloor \frac{a+b}{2}\rfloor$.
    We use \orp\ to find the last occurrence of $P_1$ before and including $c$, and $\ors$ to find the first occurrence of $P_2$ after $c$. Then we check if they are consecutive (again using \orp\ and \ors), and if it is a valid occurrence. If yes, we add it to the output.
    Then, for both $S[a,c]$ and $S[c+1,b]$, we implicitly partition into segments of size $\lfloor n_{i+1}^{1/3}\rfloor$ and find and output all valid occurrences of distance $>n_{i+1}^{1/3}$.
    Then we follow pointers to the successor nodes of the current loci to find the loci of $P_1$ and $P_2$ in the children trees $T[a,c]$ and $T[c+1,b]$ and recurse on those trees to find all consecutive occurrences of distance within $[\alpha, \beta_{i+1}]$

\end{itemize}

\paragraph{Correctness.}
At any point before we recurse on level $i$, we check all consecutive occurrences of distance $>n_{i+1}^{1/3}$ by segmenting the current substring of $S$. By the arguments of the previous section, we will find all such valid occurrences. Thus, on the subtrees of level $i+1$, we need only care about consecutive occurrences with distance in $[\alpha,\beta_{i+1}]$.

    By the properties of the induced suffix tree decomposition, a consecutive occurrence of $P_1$ and $P_2$ that is present in $T[a,b]$ will either be fully contained in $T[a,c]$, or in $T[c+1,b]$, or the occurrence of $P_1$ is the last occurrence before and including $c$ and the occurrence of $P_2$ is the first occurrence after $c$. We check the border case each time we recurse. Thus, no consecutive occurrences get lost when we recurse. If we stop the recursion, it is either because one of the loci was in a small subtree or that no valid occurrences with distance within $[\alpha, \beta_i]$ exists in $T[a,b]$. In the first case we found all valid occurrences with distance within $[\alpha, \beta_i]$ in $T[a,b]$ by the same arguments as in Section~\ref{sec:ab-existence}.  Thus, we find all valid occurrences of $P_1$ and $P_2$.
\paragraph{Time Analysis.} 
 For finding the loci, we first spend $O(|P_1|+|P_2|)$ time in the initial suffix tree $T[0,n-1]$; after that, we spend constant time each time we recurse to follow pointers.
    The rest of the time consumption is dominated by the number queries to the orthogonal range successor data structure, which we will count next.
    
    Consider the recursion part of the algorithm as a traversal of the decomposition tree, and consider the subtree of the decomposition tree we traverse. Each leaf of that subtree is a node where we stop recursing. Since we only recurse if we know there is an occurrence to be found, there are at most $O(\occ)$ leaves. Thus, we traverse at most $O(\occ \log n)$ nodes. 
    
 Each time we recurse, we spend a constant number of \ors\ and \orp\ queries to check the border cases. Additionally, we spend 
 $O(n_i^{2/3})$ such queries on each node of level $i$ that we visit in the decomposition tree: For finding the ``large" occurrences, and additionally either for reporting everything within a small subtree or doing an existence query. For finding large occurrences, there are $O(n_i^{2/3})$  segments to check. The number of orthogonal range successor queries used for existence queries or reporting within a small subtree is bounded by the number of leaves within a cluster, which is also $O(n_i^{2/3})$.
 
    Now, let $x$ be the number of decomposition tree nodes we traverse and let $l_i$, $i=1,\dots,x$, be the level of each such node. The goal is to bound $\sum_{i=1}^x \left(\frac{n}{2^{l_i}}\right)^{2/3}$.  By the argument above, $x=O(\occ \log n)$.
    Note that because the decomposition tree is binary we have that $\left(\sum_{i=1}^x \frac{1}{2^{l_i}}\right)\leq \log n$. 
The number of queries to the orthogonal range successor data structure is thus asymptotically bounded by:

    \begin{align*}
        \sum_{i=1}^x \left(\frac{n}{2^{l_i}}\right)^{2/3}&=n^{2/3}\sum_{i=1}^x \left(\frac{1}{2^{l_i}}\right)^{2/3}\cdot 1\\ &\leq n^{2/3}\left( \sum_{i=1}^x \left(\frac{1}{2^{l_i}}\right)^{\frac{2}{3}\cdot \frac{3}{2}}\right)^{2/3}\left(\sum_{i=1}^x 1^3\right)^{1/3}
        \\&= n^{2/3}\left(\sum_{i=1}^x \frac{1}{2^{l_i}}\right)^{2/3}x^{1/3}\\
        &=O(n^{2/3}\occ^{1/3}\log n)
    \end{align*}
    For the inequality, we use Hölder's inequality, which holds for all $(x_1,\dots,x_k)\in \mathbb{R}^k$ and $(y_1,\dots, y_k)\in \mathbb{R}^k$ and $p$ and $q$ both in $(1,\infty)$ such that $1/p+1/q = 1$:
    \begin{align} \sum_{i=1}^{k}|x_iy_i|\leq \left(\sum_{i=1}^k |x_i|^p\right)^{1/p}\left(\sum_{i=1}^k |y_i|^q\right)^{1/q}\label{eq:holder}
    \end{align}
    We apply (\ref{eq:holder}) with $p=3/2$ and $q=3$. 
    
Since the data structure of Zhou~\cite{DBLP:journals/ipl/Zhou16} uses $O(\log \log n)$ time per query, the total running time of the algorithm is $O(|P_1|+|P_2|+n^{2/3}\occ^{1/3}\log n\log \log n)$.
     This concludes the proof of Theorem~\ref{thm:main-result}(\ref{thm:ab-reporting}).

\section{Lower Bound}\label{sec:lower-bound}
We now prove the conditional lower bound from Theorem~\ref{thm:lower_bound} based on set intersection. We use the framework and conjectures as stated in  Goldstein~et~al.~\cite{DBLP:conf/wads/GoldsteinKLP17}. Throughout the section, let $\mathcal{I} = S_1, ,\dots, S_m$ be a collection of $m$ sets of total size $N$ from a universe $U$. The 
\emph{SetDisjointness problem} is to preprocess $\mathcal{I}$ into a compact data structure, such that given any pair of sets $S_i$ and $S_j$, we can quickly determine if $S_i \cap S_j = \emptyset$. We use the following conjecture. 
\begin{conjecture}[Strong SetDisjointness Conjecture]\label{con:strong_disjoint}
Any data structure that can answer SetDisjointness queries in $t$ query time must use $\widetilde{\Omega}\left(\frac{N^{2}}{t^2}\right)$ space.
\end{conjecture}

\subsection{SetDisjointness with Fixed Frequency}
We define the following weaker variant of the SetDisjointness problem: the \emph{$f$-FrequencySetDisjointness problem} is the SetDisjointness problem where every element occurs in precisely $f$ sets. We now show that any solution to the $f$-FrequencySetDisjointness problem implies a solution to SetDisjointness, matching the complexities up to polylogarithmic factors. 
    
    \begin{lemma} \label{lem:SDFF} Assuming the Strong SetDisjointness Conjecture, every data structure that can answer  $f$-FrequencySetDisjointness queries in time $O(N^{\delta})$, for $\delta\in[0,1/2]$, must use $\widetilde{\Omega}\left({N^{2-2\delta-o(1)}}{}\right)$ space.\end{lemma}
    \begin{proof}
 Assume there is a data structure $D$ solving the $f$-FrequencySetDisjointness problem in time $O(N^{\delta})$ and space $O\left({N^{2-2\delta-\epsilon}}{}\right)$ for constant $\epsilon$ with $0<\epsilon<1$. Let $\mathcal{I}=S_1,\dots, S_m$ be a given instance of SetDisjointness, where each $S_i$ is a set of elements from universe $U$, and assume w.l.o.g.\ that $m$ is a power of two. 
 
 Define the \emph{frequency} of an element, $f_e$, as the number of sets in $\mathcal{I}$ that contain $e$. We construct $\log m$ instances $\mathcal{I}_1,\dots, \mathcal{I}_{\log m}$ of the $f$-FrequencySetDisjointness problem. For each $j$, $1\leq j\leq \log m$, the instance $\mathcal{I}_j$ contains the following sets: 
    \begin{itemize}
        \item  For each $i\in[1,m]$ a set $S_i^j$ containing all $e\in S_i$ that satisfy $2^{j-1}\leq f_e < 2^j$;
        \item $2^{j-1}$ ``dummy sets", which contain extra copies of elements to make sure that all elements have the same frequency. That is, we add every element with $2^{j-1}\leq f_e < 2^j$ to the first $2^j-f_e$ dummy sets. These sets will not be queried in the reduction. 
    \end{itemize}
    Instance $\mathcal{I}_j$ has $O(m)$ sets and every element occurs exactly $2^j$ times.
    Further, the total number of elements in all the instances is at most $2N$.
    We now build  $f$-FrequencySetDisjointness data structures $D_j=D(\mathcal{I}_j)$ for each of the $\log m$ instances.

    To answer a SetDisjointness query for two sets $S_{i_1}$ and $S_{i_2}$, we query $D_j$ for the sets $S_{i_1}^j$ and $S_{i_2}^j$, for each $1\leq j\leq  \log m $ . If there exists a $j$ such that $S_{i_1}^j$ and $S_{i_2}^j$ are not disjoint, we output that $S_i$ and $S_j$ are not disjoint. Else, we output that they are disjoint. 
    
    
    

     If there exists $e\in S_{i_1}\cap S_{i_2}$, let $j$ be such that $2^{j-1}\leq f_e < 2^j$. Then $e\in S^j_{i_1}\cap S^j_{i_2}$, and we will correctly output that the sets are not disjoint. 
    If $S_{i_1}$ and $S_{i_2}$ are disjoint, then, since $S_{i_1}^j$ is a subset of $S_{i_1}$ and $S_{i_2}^j$ is a subset of $S_{i_2}$, the queried sets are disjoint in every instance. Thus we also answer correctly in this case.

   Let $N_j$ denote the total number of elements in $\mathcal{I}_j$. For each $j$, we have $N_j\leq 2N$ and thus $N_j^{2-2\delta-\epsilon}\leq (2N)^{2-2\delta-\epsilon}$. Thus, the space complexity is asymptotically bounded by
    \begin{align*}
        \sum_{j=1}^{\lceil\log m\rceil}N_j^{2-2\delta-\epsilon} = O(N^{2-2\delta-\epsilon}\log m ).
    \end{align*}
    Similarly, we have $N_j^{\delta}=O(N^{\delta})$ and so the time complexity is asymptotically bounded by
    \begin{align*}
        \sum_{j=1}^{\lceil\log m\rceil}N_j^{\delta} = O(N^{\delta}\log m ).
    \end{align*}
    This is a contradiction to Conjecture \ref{con:strong_disjoint}.

    \end{proof}

\subsection{Reduction to Gapped Indexing}
We can reduce  the $f$-FrequencySetDisjointness problem to $\qexistsab$ queries of the gapped indexing problem: Assume we are given an instance of the $f$-FrequencySetDisjointness problem with a total of $N$ elements. Each distinct element occurs $f$ times. Assume again w.l.o.g.\ that the number of sets $m$ is a power of two. Assign to each set $S_i$ in the instance a unique binary string $w_i$ of length $\log m$. 
Build a string $S$ as follows: Consider an arbitrary ordering $e_1,e_2,...$ of the distinct elements present in the instance. Let $\$$ be an extra letter not in the alphabet. The first $B=f\cdot\log m+f$ letters are a concatenation of $w_i\$$ of all sets $S_i$ that $e_1$ is contained in, sorted by $i$. This block is followed by $B$ copies of \$. Then, we have $B$ symbols consisting of the strings for each set that $e_2$ is contained in, again followed by $B$ copies of \$, and so on. See Figure~\ref{fig:sdff_reduction} for an example.
\begin{figure}[t]
    \centering
    \includegraphics[width=0.6\linewidth]{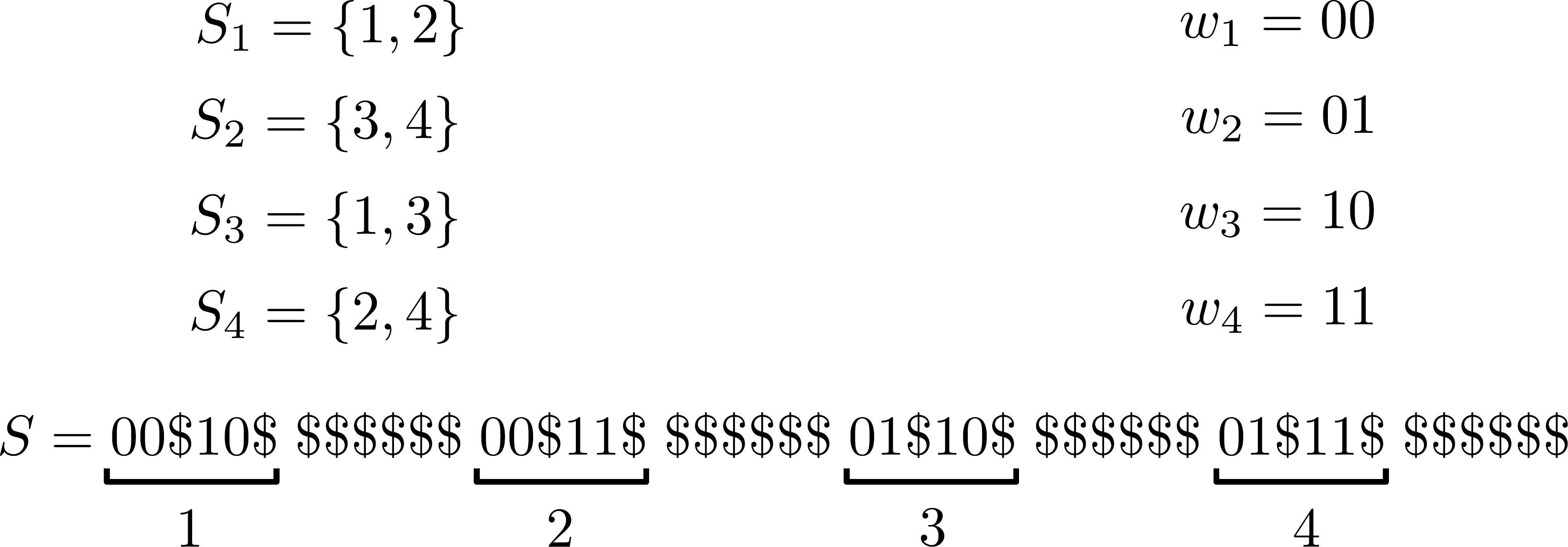}
    \caption{Instance of  $f$-FrequencySetDisjointness problem reduced to $\qexistsab$. Alphabet $\Sigma = \{0,1\}$ and fixed frequency $f = 2$, resulting in block size $B = 2 \cdot 2 + 2 = 6$.}
    \label{fig:sdff_reduction}
\end{figure}

For a query for two sets $S_i$ and $S_j$, where $i<j$, we set $P_1=w_i$ and $P_2=w_j$, $\alpha=0$, and $\beta = B$. If the sets are disjoint, then there are no occurrences which are at most $B$ apart. Otherwise $w_i$ and $w_j$ occur in the same block, and $w_j$ comes after $w_i$. 
The length of the string $S$ is $2N\log m+2N$: In the block for each element, we have $\log m+1$ letters for each of its occurrences, and it is followed by a $\$$ block of the same length. 
    
This means that if we can solve $\qexistsab$ queries in $s(n)$ space and $t(n)+O(|P_1|+|P_2|)$ time, where $n$ is the length of the string, we can solve the  $f$-FrequencySetDisjointness problem in $s(2N\log m+2N)$ space and $t(2N\log m+2N)+O(\log m)$ time.
Together with Lemma~\ref{lem:SDFF}, Theorem~\ref{thm:lower_bound} follows.

\section{Gapped Indexing for $[0, \beta]$ Gaps} \label{sec:leq-b-reporting}
In this section, we consider the special case where the queries are one sided intervals of the form $[0,\beta]$. 
    We give a data structure supporting the following tradeoffs:
    \begin{theorem}\label{thm:extension}
    Given a string of length $n$, we can
    \begin{enumerate}[(i)]
        \item construct an $O(n)$ space data structure that supports $\qexistsab(P_1,P_2,0,\beta)$ queries in $O(|P_1| + |P_2| + \sqrt{n}\log^{\epsilon} n)$ time for constant $\epsilon>0$, or
        \label{thm:b-existence-and-count}
        \item construct an $O(n\log n)$ space data structure that supports $\qcountab(P_1,P_2,0,\beta)$ and $\qreportab(P_1,P_2,0,\beta)$ queries in $O(|P_1| + |P_2| + (\sqrt{n\cdot\occ})\log \log n)$ time, where $\occ$ is the size of the output.
        \label{thm:b-existence}
    \end{enumerate}
    \end{theorem} 
    Note that since the results match (up to $\log$ factors) the best known results for set intersection, this is about as good as we can hope for. We mention here that for this specific problem, a similar tradeoff follows from the strategies used by Hon~et~al.~\cite{hon2014space}. The results from that paper include (among others) a data structure for documents such that given a query of two patterns $P_1$ and $P_2$ and a number $k$, one can output the $k$ documents with the closest occurrences of $P_1$ and $P_2$. Thus, the problem is slightly different, however, with some adjustments, the results from Theorem~\ref{thm:extension} follow (up to a $\log$ factor). We show a simple, direct solution.
    
    The data structure is a simpler version of the data structure considered in the previous sections. The main idea is that for each pair of boundary nodes $u$ and $v$, we do not have to store an array of distances, but only one number that carries all the information: the smallest distance of a consecutive occurrence of $\str{u}$ and $\str{v}$. Thus, for existence, we can cluster with $\tau=\sqrt{n}$ to achieve linear space, and we do not need to check large distances separately. For the reporting solution, we store the decomposition from Section~\ref{subsec:OSTD}, and use the matrix $M$ to decide where to recurse. In the following we will describe the details.

    \paragraph{Existence data structure.} For solving \qexistsab\ queries in this setting, we cluster the suffix tree with parameter $\tau=\sqrt{n}$. Again, we store the linear space orthogonal range successor data structure by Nekrich and Navarro~\cite{DBLP:conf/swat/NekrichN12} on the suffix array. For each pair of boundary nodes $(u,v)$, we store at $M(u,v)$ the minimum distance of a consecutive occurrence of $\str{u}$ and $\str{v}$. 
    The total space is linear. To query, we proceed similarly as in Section $\ref{sec:ab-existence}$ for the ``small distances": We find the loci of $P_1$ and $P_2$. If any of the loci is not on the spine, we check all consecutive occurrences using $\findi{.}$ resp. $\findj{.}$. If both loci are on the spine, denote $b_1,~b_2$ the lower boundary nodes of the respective clusters. Find $M(b_1,b_2)$. If $M(b_1,b_2)\leq \beta$, we can immediately return $\yes$: If a valid occurrence $(i',j')$ of $\str{b_1}$ and $\str{b_2}$ exists, then either $(i',j')$ is a consecutive occurrence of $P_1$ and $P_2$, or there exists a consecutive occurrence of smaller distance. Otherwise, that is if $M(b_1,b_2)>\beta$, all valid occurrences $(i,j)$ have the property that either $i$ is in the cluster of $\loc(P_1)$ or $j$ is in the cluster of $\loc(P_2)$, and we check all such pairs using $\findi{.}$ resp. $\findj{.}$. The running time is ${O}(|P_1| + |P_2| + \sqrt{n}\log^\epsilon n)$.
    
    \paragraph{Reporting data structure.}
    For the reporting data structure, we store the decomposition of the suffix tree as described in Section~\ref{subsec:OSTD} and the $O(n\log n)$ space orthogonal range successor data structure by Zhou \cite{DBLP:journals/ipl/Zhou16} on the suffix array. For each induced subtree of level $i$ in the decomposition, we store the existence data structure we just described.
    \paragraph{Reporting algorithm.}
The algorithm follows a similar, but simpler, recursive structure as in Section \ref{sec:ab-reporting}.
        We begin by finding the loci of $P_1$ and $P_2$.
        If either of the loci is not on a spine, we find all consecutive occurrences using $\findi{.}$ resp. $\findj{.}$, check if they are valid, report these, and terminate.
        If both loci are on a spine, we check $M(b_1,b_2)$ for the lower boundary nodes $b_1$ and $b_2$. If $M(b_1,b_2)>\beta$, all valid occurrences $(i,j)$ have the property that either $i$ is in the cluster of $\loc(P_1)$ or $j$ is in the cluster of $\loc(P_2)$. We check all such pairs using $\findi{.}$ resp. $\findj{.}$, report the valid occurrences, and terminate.
     If $M(b_1,b_2)\leq \beta$, we recurse on the children trees. That is, we check the border case and follow pointers to the loci in the children trees. 
    \paragraph{Analysis.}
    The space is $O(n\log n)$, just as in Section \ref{sec:ab-reporting}.
    
    For time analysis, we spend $O(\sqrt{\frac{n}{2^{l_i}}})$ orthogonal range successor queries on the nodes in the decomposition tree of level $l_i$ where we stop the recursion. For all other nodes we visit in the tree traversal, we only spend a constant number of queries.  
        In total, we visit $O(\occ\log (n/\occ)+\occ)$ decomposition tree nodes (by following the analysis in \cite{DBLP:journals/tcs/CohenP10}), and we spend $O(\sqrt{\frac{n}{2^{l_i}}})$ orthogonal range successor queries on $O(\occ)$ many such nodes.
    
    We use the same notation as in Section \ref{sec:ab-reporting}. By $x=O(\occ)$ we now denote the number of nodes where we stop the algorithm and output. Since each such node can be seen as a leaf in a binary tree, $\sum_{i=1}^x \frac{1}{2^{l_i}}\leq 1$. We use the Cauchy-Schwarz inequality (which is a special case of Hölders with $p=q=2$). We get as an asymptotic bound for the number of orthogonal range successor queries:
    \begin{align*}
    \sum_{i=1}^x \sqrt{\frac{n}{2^{l_i}}}&=\sqrt{n}\sum_{i=1}^x \sqrt{\frac{1}{2^{l_i}}}\cdot 1 \\
    &\leq \sqrt{n} \sqrt{\sum_{i=1}^x \frac{1}{2^{l_i}}} \sqrt{\sum_{i=1}^x 1}\\
    &\leq \sqrt{n x} = O(\sqrt{n\cdot\occ}).
    \end{align*}
    Note that since $\occ\log (n/\occ)=O(\occ\sqrt{n/\occ})=O(\sqrt{n\cdot \occ})$, this brings the total number of orthogonal range successor queries to  $O(\occ + \sqrt{n\cdot\occ})$.
    Using the data structure by Zhou~\cite{DBLP:journals/ipl/Zhou16}, the time bound from Theorem~\ref{thm:extension} follows.


\section{Conclusion}
    We have considered the problem of gapped  indexing  for consecutive occurrences.
   We have given a linear space data structure that can count the number of such occurrences. For the reporting problem, we have given a near-linear space data structure. The running time for both includes an $O(n^{2/3})$ term, which forms a gap of $O(n^{1/6})$ to the conditional lower bound of $O(\sqrt{n})$. Thus, the most obvious open question is whether we can close this gap, either by improving the data structure or finding a stronger lower bound. 
   
   Further, we have used the property that there can only be few consecutive occurrences of large distances. Thus, our solution cannot be easily extended to finding \emph{all} pairs of occurrences with distance within the query interval. An open question is if it is possible to get similar results for that problem. Lastly, document versions of similar problems have concerned themselves with finding all documents that contain $P_1$ and $P_2$ or the top-$k$ of smallest distance; conditional lower bounds for these problems are also known. It would be interesting to see if any of these results be extended to finding all documents that contain a (consecutive) occurrence of $P_1$ and $P_2$ that has a distance within a query interval.

\bibliographystyle{plainurl}
\bibliography{References}

\end{document}